\documentclass[submission,copyright,creativecommons]{eptcs}
\providecommand{\event}{GANDALF 2012}

\usepackage{amsmath}
\usepackage{amssymb}
\usepackage{amsthm}
\usepackage{mathrsfs}
\usepackage{xcolor}
\usepackage{enumitem}
\usepackage{lscape}
\usepackage{mathtools}
\usepackage{multirow}
\usepackage{tikz-qtree}
\usepackage{pgf,tikz}
\usepackage{hyperref}
\usepackage{multicol}
\usepackage{mathtools}
\usepackage{subfig}

\newtheorem{Rule}{Rule}{\bfseries}{\itshape}
\newtheorem{procedure}{Procedure}{\bfseries}{\itshape}
\newtheorem{definition}{Definition}{\bfseries}{\itshape}
\newtheorem{example}{Example}{\bfseries}{\itshape}
\newtheorem{theorem}{Theorem}{\bfseries}{\itshape}
\newtheorem{corollary}{Corollary}{\bfseries}{\itshape}
\newtheorem{lemma}{Lemma}{\bfseries}{\itshape}

\newcommand{\bbN}{\mathbb{N}}

\newcommand{\calA}{\mathcal{A}}
\newcommand{\calB}{\mathcal{B}}

\newcommand{\calF}{\mathcal{F}}

\newcommand{\calT}{\mathcal{T}}

\newcommand{\scrL}{\mathscr{L}}

\newcommand{\cat}{\circ}

\makeatletter
\providecommand\barcirc{\mathpalette\@barred\cat}
\def\@barred#1#2{\ooalign{\hfil$\sharp1-$\hfil\cr\hfil$\sharp1#2$\hfil\cr}}

\makeatother

\newcommand{\LS}{\mathit{LS}}

\newcommand{\TOP}{\mathit{TOP}}
\newcommand{\ITS}{\mathit{ITS}}

\DeclareMathOperator{\ch}{\mathit{ch}}
\DeclareMathOperator{\Inf}{\mathit{Inf}}
\DeclareMathOperator{\Cover}{\mathit{Cover}}
\DeclareMathOperator{\Mini}{\mathit{Mini}}

\DeclareMathOperator{\STS}{\mathit{STS}}
\DeclareMathOperator{\RSTS}{\mathit{\mu STS}}

\DeclareMathOperator{\red}{\mathit{red}}
\DeclareMathOperator{\green}{\mathit{green}}
\DeclareMathOperator{\yellow}{\mathit{yellow}}

\DeclareMathOperator{\set}{\mathit{set}}
\DeclareMathOperator{\bucket}{\mathit{bucket}}

\DeclareMathOperator{\head}{\mathit{head}}

\title{Can Nondeterminism Help Complementation?\footnote{The research was supported by NSF grants 0954132 and 1143830.}}
\author{Yang Cai
\institute{MIT CSAIL \\
The Stata Center, 32-G696 \\
Cambridge, MA 02139 USA}
\email{ycai@csail.mit.edu}
\and
Ting Zhang
\institute{Iowa State University \\
226 Atanasoff Hall \\
Ames, IA 50011 USA}
\email{tingz@cs.iastate.edu}
}
\def\titlerunning{Can Nondeterminism Help Complementation?}
\def\authorrunning{Y. Cai, \& T. Zhang}

\begin{document}
\maketitle
\date{}
\thispagestyle{empty}

\begin{abstract}
Complementation and determinization are two fundamental notions in automata theory. The close relationship between the two has been well observed in the literature. In the case of nondeterministic finite automata on finite words (NFA), complementation and determinization have the same state complexity, namely $\Theta(2^{n})$ where $n$ is the state size. The same similarity between determinization and complementation was found for B\"{u}chi automata, where both operations were shown to have $2^{\Theta(n \lg n)}$ state complexity. An intriguing question is whether there exists a type of $\omega$-automata whose determinization is considerably harder than its complementation. In this paper, we show that for all common types of $\omega$-automata, the determinization problem has the same state complexity as the corresponding complementation problem at the granularity of $2^{\Theta(\cdot)}$. In particular, we show a determinization construction for Streett automata with state complexity bounded by $2^{12n}(0.37n)^{n^{2}+8n}$.
\end{abstract}

\section{Introduction}
\label{sec:introduction}

Automata on infinite words ($\omega$-automata) have wide applications in synthesis and verification of reactive concurrent systems. Complementation and determinization are two fundamental notions in automata theory. The complementation problem is to construct, given an $\omega$-automaton $\calA$, an $\omega$-automaton $\calB$ that recognizes the complementary language of $\calA$. Provided that a given $\calA$ is nondeterministic, the determinization problem is to construct a deterministic $\omega$-automaton $\calB$ that recognizes the same language as $\calA$ does.

The close relationship between determinization and complementation has been well observed in the literature (see~\cite{Choueka/74/Theories,Vardi/07/Buchi} for discussions). A deterministic $\omega$-automaton can be trivially complemented by dualizing its acceptance condition. As a result, a lower bound on complementation applies to determinization while an upper bound on determinization applies to complementation. Besides easily rendering complementation, determinization is crucial in decision problems for tree temporal logics, logic games and system synthesis. For example, using game-theoretical semantics~\cite{Gurevich+Harrington/82/Trees}, complementation of $\omega$-tree automata ($\omega$-tree complementation) not only requires complementation of $\omega$-automata, but also requires the complementary $\omega$-automata be deterministic (called co-determinization). Therefore, lowering the cost of $\omega$-determinization also improves the performance of $\omega$-tree complementation.

In the case of nondeterministic finite automata on finite words (NFA), complementation and determinization have the same state complexity, namely $\Theta(2^{n})$ where $n$ is the state size. The same similarity between determinization and complementation was found for B\"{u}chi automata, where both operations were shown to have $2^{\Theta(n \lg n)}$ state complexity. An intriguing question is whether there exists a type of $\omega$-automata whose determinization is considerably harder than its complementation. In this paper, we show that for all common types of $\omega$-automata, the determinization problem has the same state complexity as the corresponding complementation problem at the granularity of $2^{\Theta(\cdot)}$.

\paragraph{Related Work.} B\"{u}chi started the theory of $\omega$-automata. The original $\omega$-automata used to establish the decidability of S1S~\cite{Buchi/66/Decision} are now referred to as B\"{u}chi automata. Shortly after B\"{u}chi's work, McNaughton proved a fundamental theorem in the theory of $\omega$-automata, that is, $\omega$-regular languages (the languages that are recognized nondeterministic B\"{u}chi automata) are exactly those recognizable by the deterministic version of a type of $\omega$-automata, now referred to as Rabin automata\footnotemark~\cite{McNaughton/66/Testing}.

\footnotetext{\cite{McNaughton/66/Testing} used Muller condition which, however, was converted from Rabin condition.}

The complexity of McNaughton's B\"{u}chi determinization is double exponential. In 1988, Safra proposed a type of tree structures (now referred to as Safra trees) to obtain a B\"{u}chi determinization that converts a nondeterministic B\"{u}chi automaton of state size $n$ to an equivalent deterministic Rabin automaton of state size $2^{O(n\lg n)}$ and index size $O(n)$~\cite{Safra/88/Complexity}. Safra's construction is essentially optimal as the lower bound on state size for B\"{u}chi complementation is $2^{\Omega(n\lg n)}$~\cite{Michel/88/Complementation,Loding/99/Optimal}. Later, Safra generalized the B\"{u}chi determinization to a Streett determinization which, given a nondeterministic Streett automaton of state size $n$ and index size $k$, produces an equivalent deterministic Rabin automaton of state size $2^{O(nk\lg nk)}$ and index size $O(nk)$~\cite{Safra/92/Exponential}. A variant B\"{u}chi determinization using a similar tree structure was proposed by Muller and Schupp~\cite{Muller+Schupp/95/Simulating}. Recently, Piterman improved both of Safra's determinization procedures with a node renaming scheme~\cite{Piterman/06/From}. Piterman's constructions are more efficient than Safra's, though the asymptotical bounds in terms of $2^{O(\cdot)}$ are the same. A big advantage of Piterman's constructions, however, is to output deterministic parity automata, which is easier to manipulate than deterministic Rabin automata. For example, there exists no efficient procedure to complement deterministic Rabin automata to B\"{u}chi automata~\cite{Safra+Vardi/89/Omega}, while such complementation is straightforward and efficient for deterministic parity automata.

In~\cite{Cai+Zhang+Luo/09/Improved,Cai+Zhang/11/Tight+Lower,Cai+Zhang/11/Tight+Upper} we established tight bounds for the complementation of $\omega$-automata with rich acceptance conditions, namely Rabin, Streett and parity automata. The state complexities of the corresponding determinization problems, however, are yet to be settled. In particular, a large gap exists between the lower and upper bounds for Streett determinization. A Streett automaton can be viewed as a B\"{u}chi automaton $\langle Q, \Sigma, Q_{0}, \Delta, \calF\rangle$ except that the acceptance condition $\calF=\langle G, B\rangle_{I}$, where $I=[1..k]$ for some $k$ and $G, B: I \to 2^{Q}$, comprises $k$ pairs of \emph{enabling sets} $G(i)$ and \emph{fulfilling sets} $B(i)$. A run is accepting if for every $i \in I$, if the run visits $G(i)$ infinitely often, then so does it to $B(i)$. Therefore, Streett automata naturally express \emph{strong fairness} conditions that characterize meaningful computations~\cite{Francez/84/Generalized,Francez/1986/Fairness}. Moreover, Streett automata can be exponentially more succinct than B\"{u}chi automata in encoding infinite behaviors of systems~\cite{Safra+Vardi/89/Omega}. As a results, Streett automata have an advantage in modeling the behaviors of concurrent and reactive systems. For Streett determinization, the gap between current lower and upper bounds is huge: the lower bound is $2^{\Omega(n^{2} \lg n)}$~\cite{Cai+Zhang/11/Tight+Lower} and the upper bound is $2^{O(nk\lg nk)}$~\cite{Safra/92/Exponential,Piterman/06/From} when $k$ is large (say $k=\omega(n)$).

In this paper we focus on Streett determinization. We show a Streett determinization whose state complexity matches the lower bounds established in~\cite{Cai+Zhang/11/Tight+Lower}. More precisely, our construction has state complexity $2^{O(n\lg n+nk\lg k)}$ for $k=O(n)$ and $2^{O(n^{2} \lg n)}$ for $k=\omega(n)$. We note that this improvement is not only meant for large $k$. When $k= O(\log n)$, the difference between $2^{O(n\lg n+nk\lg k)}$ and $2^{O(nk\lg nk)}$ is already substantial. In fact, no matter how large $k$ is, our state complexity is always bounded by $2^{12n}(0.37n)^{n^{2}+8n}$ while the current best bound for $k=n-1$is $(12)^{n^{2}}n^{3n^{2}+n}$~\cite{Piterman/06/From}.

The phenomenon that determinization and complementation have the same state complexity does not stop at Streett automata; we also show that this phenomenon holds for generalized B\"{u}chi automata, parity automata and Rabin automata. This raises a very interesting question: do determinization and complementation always ``walk hand in hand''?  Although the exact complexities of complementation and determinization for some or all types of $\omega$-automata could be different, the ``coincidence'' at the granularity of $2^{\Theta(\cdot)}$ is already intriguing\footnotemark.

\footnotetext{Recent work by Colcombet and Zdanowski, and by Schewe showed that the state complexity of determinization of B\"{u}chi automata on alphabets of \emph{unbounded} size is between $\Omega((1.64n)^{n})$~\cite{Colcombet+Zdanowski/09/Tight} and $O((1.65n)^{n})$~\cite{Schewe/09/Tighter}, which is strictly higher than the state complexity of B\"{u}chi complementation which is between $\Omega(L(n))$~\cite{Yan/06/Lower} and $O(n^{2}L(n))$~\cite{Schewe/09/Buchi} (where $L(n) \approx (0.76n)^{n})$). However, Schewe's determinization construction produces Rabin automata with exponentially large index size, and it is not known whether Colcombet and Zdanowski's lower bound result can be generalized to B\"{u}chi automata on conventional alphabets of fixed size.}

\paragraph{Our Approaches.}
Our improved construction for Streett determinization bases on two ideas. The first one is what we have exploited in obtaining tight upper bounds for Streett complementation~\cite{Cai+Zhang/11/Tight+Upper}, namely, the larger the size of Streett index size, the higher correlations in the runs of Streett automata. We used two tree structures: $\ITS$ (\emph{Increasing Tree of Sets}) and $\TOP$ (\emph{Tree of Ordered Partitions}) to characterize those correlations. We observed that there is a similarity between $\TOP$ and Safra trees for B\"{u}chi determinization~\cite{Safra/88/Complexity}. As Safra trees for Streett determinization are generalization of those for B\"{u}chi determinization, we conjectured that $\ITS$ should have a role in improving Streett determinization. Our study confirmed this expectation; using $\ITS$ we can significantly reduce the size of Safra trees for Streett determinization.

The second idea is a new naming scheme. Bounding the size of Safra trees alone cannot bring down the state complexity when the Streett index is small (i.e., $k=O(n)$), because the naming cost becomes a dominating factor in this case. Naming is an integral part of Safra trees. Every node in a Safra tree is associated with a name, which is used to track changes of the node between the tree (state) and its successors. The current name allocation is a retail-style strategy; when a new node is created, a name from the pool of unused names is selected arbitrarily and assigned, and when a node is removed, its name is recycled to the pool. In contrast, our naming scheme is more like wholesale; the name space is divided into even blocks and every block is allocated at the same time. When a branch is created, an unused block is assigned to it, and when a branch is changed, the corresponding block is recycled.

\paragraph{Paper Organization.}
Section~\ref{sec:preliminaries} presents notations and basic terminology in automata theory. Section~\ref{sec:safra-determinization} introduces Safra's construction for B\"{u}chi and Streett determinization. Section~\ref{sec:improved-Streett-determinization} presents our construction for Streett determinization. Section~\ref{sec:complexity} establishes tight upper bounds for the determinization of Streett, generalized B\"{u}chi, parity, and Rabin automata. Section~\ref{sec:conclusion} concludes with some discussion on future work. Due to space limit, most of the proofs are not included, but they can be found in the full version of this paper on authors' websites.


\section{Preliminaries}
\label{sec:preliminaries}

\paragraph{Basic Notations.}
Let $\bbN$ denote the set of natural numbers. We write $[i..j]$ for $\{ k \in \bbN \, \mid \, i \le k \le j\}$, $[i..j)$ for $[i..j-1]$, and $[n]$ for $[0..n)$. For an infinite sequence $\varrho$, we use $\varrho(i)$ to denote the $i$-th component for $i \in \bbN$. For a finite sequence $\alpha$, we use $|\alpha|$ to denote the length of $\alpha$, $\alpha[i]$ ($i \in [1..|\alpha|]$) to denote the object at the $i$-th position, and $\alpha[i..j]$ (resp. $\alpha[1..j)$) to denote the subsequence of $\alpha$ from position $i$ to position $j$ (resp. $j-1$). We reserve $n$ and $k$ as parameters of a determinization instance ($n$ for state size and $k$ for index size). Define $I=[1..k]$.


\paragraph{Automata and Runs.}
An $\omega$-automaton is a tuple $\calA=(\Sigma, Q, Q_{0}, \Delta, \calF)$ where $\Sigma$ is an alphabet, $Q$ is a finite set of states, $Q_{0} \subseteq Q$ is a set of initial states, $\Delta \subseteq Q \times \Sigma \times Q$ is a set of transition relations, and $\calF$ is an acceptance condition. A \emph{finite run} of $\calA$ from state $q$ to state $q'$ over a finite word $w$ is a sequence of states $\varrho=\varrho(0)\varrho(1)\cdots\varrho(|w|)$ such that $\varrho(0)=q, \varrho(|w|)=q'$ and $\langle\varrho(i),w(i),\varrho(i+1)\rangle \in\Delta$ for all $i \in [|w|]$. An infinite word ($\omega$-words) over $\Sigma$ is an infinite sequence of letters in $\Sigma$. A \emph{run} $\varrho$ of $\calA$ over an $\omega$-word $w$ is an infinite sequence of states in $Q$ such that $\varrho (0)\in Q_{0}$ and, $\langle\varrho(i),w(i),\varrho(i\!+\!1)\rangle \in \Delta$ for $i \in \bbN$. Let $\Inf(\varrho)$ be the set of states that occur infinitely many times in $\varrho$. An automaton accepts $w$ if $\varrho$ satisfies $\calF$, which usually is defined as a predicate on $\Inf(\varrho)$. By $\scrL(\calA)$ we denote the set of $\omega$-words accepted by $\calA$.

\paragraph{Acceptance Conditions and Types.}
$\omega$-automata are classified according their acceptance conditions. Below we list automata of common types. Let $G, B$ be two functions $I \to 2^{Q}$.
\begin{itemize}
\item \emph{Generalized B\"{u}chi}: $\langle B \rangle_{I}$: $\forall i \in I, \Inf(\varrho) \cap B(i) \neq \emptyset$.
\item \emph{B\"{u}chi}: $\langle B \rangle_{I}$ with $I=\{1\}$ (i.e., $k=1$).
\item \emph{Streett}, $\langle G,B\rangle_{I}$: $\forall i \in I$, $\Inf(\varrho) \cap G(i)\neq\emptyset \to \Inf(\varrho)\cap B(i)\neq\emptyset$.
\item \emph{Parity}, $\langle G,B\rangle_{I}$ with $B(1) \subset G(1) \subset \cdots \subset B(k) \subset G(k)$.
\item \emph{Rabin}, $[G,B]_{I}$: $\exists i \in I$, $\Inf(\varrho)\cap G(i)\neq\emptyset \wedge \Inf(\varrho)\cap B(i)=\emptyset$.
\end{itemize}
For simplicity, we denote a B\"{u}chi condition by $F$ (i.e., $F=B(1)$) and call it the final set of $\calA$. Note that B\"{u}chi, generalized B\"{u}chi and parity automata are all subclasses of Streett automata. For a Streett condition $\langle G,B\rangle_{I}$, if there exist $i,i' \in I$, $B(i)=B(i')$, then we can simplify the condition by replacing both $\langle G(i), B(i)\rangle$ and $\langle G(i'), B(i')\rangle$ by $\langle G(i) \cup G(i'), B(i) \rangle$. For this reason, for every Streett condition $\langle G,B\rangle_{I}$, we assume that $B$ is injective and hence $k=|I| \le 2^{n}$.

\paragraph{Trees.}
A tree is a set $V \subseteq \bbN^{*}$ such that if $v \cdot i \in V$ then $v \in V$ and $v \cdot j \in V$ for $j \le i$. In this paper we only consider finite trees. Elements in $V$ are called \emph{nodes} and $\epsilon$ is called the \emph{root}. Nodes $v \cdot i$ are \emph{children} of $v$ and they are \emph{siblings} to one another. The set of $v$'s children is denoted by $\ch(v)$. A node is a leaf if it has no children. Given an alphabet $\Sigma$, a \emph{$\Sigma$-labeled tree} is a pair $\langle V, L \rangle$, where $V$ is a tree and $L: V \to \Sigma$ assigns a letter to each node in $V$. We refer to $v \in V$ as \emph{$V$-value} (\emph{structural value}) and $L(v)$ as \emph{$L$-value} (\emph{label value}) of $v$. 
\section{Safra's Streett Determinization}
\label{sec:safra-determinization}

In this section we introduce Safra's constructions for Streett determinization.

We start with the idea behind B\"{u}chi determinization~\cite{Safra/88/Complexity}. Let $\calA=\langle Q, Q_{0}, \Sigma, \Delta, F \rangle$ be a nondeterministic B\"{u}chi automaton. By the standard subset construction~\cite{Rabin+Scott/59/Finite}, it is not hard to build a deterministic $\omega$-automaton $\calB$ such that if a run $\varrho=q_{0}, q_{1}, \ldots \in Q^{\omega}$ of $\calA$ over an infinite word $w$ is accepting, then so is the run $\tilde{\varrho}=\tilde{q}_{0}, \tilde{q}_{1}, \ldots \in (2^{Q})^{\omega}$ of $\calB$ over $w$, and $q_{i} \in \tilde{q}_{i}$ for every $i \ge 0$. The hard part is to also guarantee that an accepting run $\tilde{\varrho}$ of $\calB$ over $w$ induces an accepting run $\varrho$ of $\calA$ over $w$. Let $S \stackrel{F}{\hookrightarrow} S'$ denote that for every state $q' \in S'$, there exists a state $q \in S$ and a finite run $\rho$ of $\calA$ such that $\rho$ goes from $q$ to $q'$ and visits $F$. The key idea in many determinization constructions~\cite{McNaughton/66/Testing,Safra/88/Complexity,Safra/92/Exponential,Muller+Schupp/95/Simulating,Piterman/06/From} relies on the following lemma, which itself is a consequence of K\"{o}nig's lemma.

\begin{lemma}[\cite{McNaughton/66/Testing,Safra/88/Complexity}]
\label{lem:Buchi-connecting-dots}
Let $\tilde{\varrho}=\tilde{q}_{0}, \tilde{q}_{1}, \ldots \in (2^{Q})^{\omega}$ and $(S_{i})_{\omega}=S_{0}, S_{1}, \ldots$ an infinite subsequence of $\tilde{\varrho}$ such that for every $i \ge 0$, $S_{i} \stackrel{F}{\hookrightarrow} S_{i+1}$, then there is an accepting run $\varrho=q_{0}, q_{1}, \ldots \in Q^{\omega}$ of $\calA$ such that $q_{i} \in \tilde{q}_{i}$ for every $i \ge 0$.
\end{lemma}

The idea was generalized to Streett determinization~\cite{Safra/92/Exponential}. Let $\calA=\langle Q, Q_{0}, \Sigma, \Delta, \langle G,B\rangle_{I} \rangle$ be a Streett automaton. An index set $J \subseteq I$ serves as a \emph{witness set} for an accepting run $\varrho$ of $\calA$ in the following sense: for every $j$ in $J$, $\varrho$ visits $B(j)$ infinitely often while for every $j \in I \setminus J$, $\varrho$ visits $G(j)$ only finitely many times. We call indices in $I \setminus J$ \emph{negative obligations} and indices in $J$ \emph{positive obligations}. It is easily seen that a run $\varrho$ is accepting if and only if $\varrho$ admits a witness set $J \subseteq I$. We say that a finite run $\rho$ \emph{fulfills} $J$ if for every $j$ in $J$, $\rho$ visits $B(j)$, but for every $j \in I \setminus J$, $\rho$ does not visit $G(j)$. By $S \stackrel{J}{\hookrightarrow} S'$ we mean that every state in $S'$ is reachable from a state in $S$ via a finite run $\rho$ that fulfills $J$. We have the following lemma analogous to Lemma~\ref{lem:Buchi-connecting-dots}.

\begin{lemma}[\cite{Safra/92/Exponential}]
\label{lem:Streett-connecting-dots}
Let $\tilde{\varrho}=\tilde{q}_{0}, \tilde{q}_{1}, \ldots \in (2^{Q})^{\omega}$ and $(S_{i})_{\omega}=S_{0}, S_{1}, \ldots$ an infinite subsequence of $\tilde{\varrho}$ such that for every $i \ge 0$, $S_{i} \stackrel{J}{\hookrightarrow} S_{i+1}$, then there is an accepting run $\varrho=q_{0}, q_{1}, \ldots \in Q^{\omega}$ of $\calA$ such that $q_{i} \in \tilde{q}_{i}$ for every $i \ge 0$.
\end{lemma}

The ingenuity of Safra construction is to efficiently organize information in a type of tree structures, now referred to as Safra trees. In a Safra tree for Streett determinization, each node is labeled by a witness set (index label) and a state set (state label). We simply say that a node \emph{contains} the states in its state label. The standard subset construction is carried out on all nodes in parallel. Each node $v$ with witness set $J$ tracks runs that fulfill $J$. We say that $v$ turns $\green$ when $J$ is fulfilled, and we have $S \stackrel{J}{\hookrightarrow} S'$ where $S$ and $S'$ are the state labels of $v$ at two consecutive moments of turning green. If $v$ turns green infinitely often, then we have a desired $(S_{i})_{\omega}$ as stated in Lemma~\ref{lem:Streett-connecting-dots}. We illustrate this idea by an example. Let $k=3$ and $I=[1..3]$.

The initial state in the deterministic automaton $\calB$ is a root $v$ labeled with $I$ and $Q_{0}$. The obligation of $v$ is to detect runs that fulfills $I$. Once a run visits $B(3)$ (fulfilling a positive obligation), the run moves to a new child, waiting to visit $B(2)$, and once the run visits $B(2)$, the run moves to anther new child, waiting to visit $B(1)$, and so on. Technically speaking, states are moved around, which induces moving of runs.

This \emph{sequential sweeping} can be stalled because some runs in $v$, from some point on, could never visit $B(3)$. So $v$ should spawn a child $v_{3}$ with the witness set $I \setminus \{3\}=\{1,2\}$, to detect runs that fulfill $\{1,2\}$ (see  in Figure~\ref{fig:STS-witness-set-label}). A run in $v_{3}$ should never visit $G(3)$, its negative obligation. Or the run will be reset (the exact meaning of reset is shown in Step~\eqref{en:STS-expansion-nonleaf-Gs}). If a run fulfills $I$, then the run moves into a special child that also has $I$ as its witness set. If all children of $v$ are special, then we say $v$ turns $\green$. It is not hard to verify that if $S$ and $S'$ are the state sets of $v$ at any two consecutive moments when $v$ turns green, then we have $S \stackrel{J}{\hookrightarrow} S'$. But $v$ may never turn green because not all runs fulfill $I$. Thus the special children of $v$ also have $I$ as their witness sets, for they also spawn children, behaving just as $v$. This spawning should be recursively applied until we arrive at leaves whose witness sets are singletons. We refer the reader to~\cite{Safra/92/Exponential,Schwoon/2002/Determinization,Piterman/06/From} for a detailed exposition of this spawning and sweeping process. Figure~\ref{fig:witness-set-illustration} shows part of a Safra tree for Streett determinization in two equivalent representations (with respect to witness sets). The representation shown in Figure~\ref{fig:STS-index-label} is used in Definition~\ref{def:STS}.

\begin{figure}[t!]
\centering
\subfloat[][]{%
\label{fig:STS-witness-set-label}%
\begin{tikzpicture}[scale=0.60]
\Tree [.$v:\{1,2,3\}$
        [.$v_{3}:\{1,2\}$
          [.$v_{32}:\{1\}$ ]
          [.$\{2\}$ ]
          [.$\{1,2,3\}$ ]
        ]
        [.$v_{2}:\{1,3\}$
          [.$\{1\}$ ]
          [.$\{3\}$ ]
          [.$\{1,2,3\}$ ]
        ]
        [.$v_{1}:\{2,3\}$
          [.$\{2\}$ ]
          [.$\{3\}$ ]
          [.$\{1,2,3\}$ ]
        ]
        [.$v_{0}:\{1,2,3\}$
          [.$\{1,2\}$ ]
          [.$\{1,3\}$ ]
          [.$\{2,3\}$ ]
          [.$\{1,2,3\}$ ]
        ]
      ]
\end{tikzpicture}
}%
\hspace{10pt}%
\subfloat[][]{%
\label{fig:STS-index-label}%
\begin{tikzpicture}[scale=0.60]
\Tree [.$v:0$
        [.$v_{3}:3$
          [.$v_{32}:2$ ]
          [.$1$ ]
          [.$0$ ]
        ]
        [.$v_{2}:2$
          [.$3$ ]
          [.$1$ ]
          [.$0$ ]
        ]
        [.$v_{1}:1$
          [.$3$ ]
          [.$2$ ]
          [.$0$ ]
        ]
        [.$v_{0}:0$
          [.$3$ ]
          [.$2$ ]
          [.$1$ ]
          [.$0$ ]
        ]
      ]
\end{tikzpicture}
}%
\caption[Two representations of a Safra tree for Streett determinization]{Two equivalent representations of witness sets.
In~\subref{fig:STS-witness-set-label}, nodes are explicitly labeled with witness sets. In~\subref{fig:STS-index-label}, witness sets are implicitly represented by index labels. In this representation, the witness set of a node is the set of indices that do not appear on the path from the root to the node.}
\label{fig:witness-set-illustration}
\end{figure}

A key ingredient in Safra's construction is to assign names and colors to nodes, in order to track changes on nodes to identify those that turn green infinitely many times. In this paper, use three colors: $\green$, $\red$ and $\yellow$.
\begin{definition}[Safra Trees for Streett Determinization ($\STS$)]
\label{def:STS}
A \emph{Safra tree for Streett determinization} ($\STS$) is a labeled tree $\langle V, L \rangle$ with $L=\langle L_{n}, L_{s}, L_{c}, L_{h}\rangle$ where
\begin{enumerate}[label=\ref{def:STS}.\arabic*,ref=\ref{def:STS}.\arabic*]
\item\label{en:STS-name} $L_{n}: V \to [1..nk]$ assigns each node a unique name.
\item\label{en:STS-set-label} $L_{s}: V \to 2^{Q}$ assigns each node a subset of $Q$ such that for every node $v$, $L_{s}(v)=\cup_{v' \in \ch(v)} L_{s}(v')$ and $L_{s}(v') \cap L_{s}(v'') = \emptyset$ for every two distinct $v', v'' \in \ch(v)$.
\item\label{en:STS-color} $L_{c}: V \to \{\green, \red, \yellow\}$ assigns each node a color.
\item\label{en:STS-index-label} $L_{h}: V \to I \cup \{0\}$ assigns each node an index in $I \cup \{0\}$. For a node $v$, let $L^{\to}_{h}(v)$ denote the sequence of $I$-elements from the root to $v$ with $0$ excluded and $L^{\set}_{h}(v)$ the set of $I$-elements occurring in $L^{\to}_{h}(v)$. We require that for every node $v$, there is no repeated index in the sequence $L^{\to}_{h}(v)$.
\end{enumerate}
\end{definition}
More precisely, an $\STS$ is a labeled and ordered tree; nodes are partially ordered by \emph{older-than} relation. In the tree transformation (see Procedure~\ref{pro:Safra-Streett-determinization}), a newly added node is considered \emph{younger} than its existing siblings. We choose not to define the ordering formally as the meaning is clear from the context. Also, for the sake of presentation clarity, we separate the following naming and coloring convention from the core construction (Step~\ref{en:STS-delta} of Procedure~\ref{pro:Safra-Streett-determinization}).
\begin{Rule}[Naming and Coloring Convention]
\label{rule:ST-naming-coloring}
(1) newly created nodes are marked $\red$, (2) nodes whose all descendants are removed are marked $\green$, (3) nodes are marked $\yellow$ unless they have been marked $\red$ or $\green$, (4) nodes with the empty state label are removed, and (5) when a node is created, a name from the pool of unused name is selected and assigned to the node; when a node is removed, its name is recycled back to the pool.
\end{Rule}

\begin{procedure}[Streett Determinization~\cite{Safra/92/Exponential}]
\label{pro:Safra-Streett-determinization}
Let $\calA=\langle Q, Q_{0}, \Sigma, \Delta, \langle G, B \rangle_{I} \rangle$ ($I=[1..k]$) be a nondeterministic Streett automaton. The following procedure outputs a deterministic Rabin automaton $\calB=\langle \tilde{Q}, \tilde{q}_{0}, \Sigma, \tilde{\Delta}, [\tilde{G}, \tilde{B}]_{\tilde{I}} \rangle$ ($\tilde{I}=[1..nk]$) such that
\begin{enumerate}[label=\ref{pro:Safra-Streett-determinization}.\arabic*,ref=\ref{pro:Safra-Streett-determinization}.\arabic*]
\item \label{en:STS} $\tilde{Q}$ is the set of $\STS$.
\item \label{en:STS-initial} $\tilde{q}_{0} \in \tilde{Q}$ is the tree with just the root $v$ such that $L_{n}(v)=1$, $L_{s}(v)=Q_{0}$, $L_{c}(v)=\red$ and $L_{h}(v)=0$.
\item \label{en:STS-delta} $\tilde{q}'=\tilde{\Delta}(\tilde{q})$ is defined the $\STS$ obtained by applying the following transformation rule to $\tilde{q}$.
\begin{enumerate}[label=\ref{en:STS-delta}.\arabic*,ref=\ref{en:STS-delta}.\arabic*]
\item \label{en:STS-subset-construction} \emph{Subset Construction}: for each node $v$ in $\tilde{Q}$, update state label $L_{s}(v)$ to $\Delta(L_{s}(v))$.
\item \label{en:STS-expansion} \emph{Expansion}. Apply the following transformations \emph{downwards} from the root:
\begin{enumerate}[label=\ref{en:STS-expansion}.\arabic*,ref=\ref{en:STS-expansion}.\arabic*]
\item \label{en:STS-expansion-leaf-empty} If $v$ is a leaf with $L^{\set}_{h}(v) = I$, then stop.
\item \label{en:STS-expansion-leaf-nonempty} If $v$ is a leaf with $L^{\set}_{h}(v) \not = I$, then add a child $v'$ to $v$ such that $L_{s}(v')=L_{s}(v)$ and $L_{h}(v')=\max(I \setminus L^{\set}_{h}(v))$.
\item \label{en:STS-expansion-nonleaf} If $v$ is a node with $j$ children $v_{1}, \ldots, v_{j}$, then let $i_{1}, \ldots, i_{j}$ be the corresponding index labels, and consider the following two cases for each $j' \in [1..j]$.
\begin{enumerate}[label=\ref{en:STS-expansion-nonleaf}.\arabic*,ref=\ref{en:STS-expansion-nonleaf}.\arabic*]
\item \label{en:STS-expansion-nonleaf-Bs} If $L_{s}(v_{j'}) \cap B(i_{j'}) \not = \emptyset$, then add a child $v'$ to $v$ with $L_{s}(v')=L_{s}(v_{j'}) \cap B(i_{j'})$ and $L_{h}(v')=\max([0..i_{j'}) \cap ((I \cup \{0\}) \setminus L^{\set}_{h}(v)))$, and remove the states in $L_{s}(v_{j'}) \cap B(i_{j'})$ from $v_{j'}$ and all the descendants of $v_{j'}$.
\item \label{en:STS-expansion-nonleaf-Gs} If $L_{s}(v_{j'}) \cap B(i_{j'}) = \emptyset$ and $L_{s}(v_{j'}) \cap G(i_{j'}) \not = \emptyset$, then add a child $v'$ to $v$ with $L_{s}(v')=L_{s}(v_{j'}) \cap G(i_{j'})$ and $L_{h}(v')=L_{h}(v)$, and remove the states in $L_{s}(v_{j'}) \cap G(i_{j'})$ from $v_{j'}$ and all the descendants of $v_{j'}$.
\end{enumerate}
\end{enumerate}
\item \label{en:STS-h-merge} \emph{Horizontal Merge}. For any state $q$ and any two siblings $v$ and $v'$ such that $q \in L_{s}(v) \cap L_{s}(v')$, if $L_{h}(v) < L_{h} (v')$, or $L_{h}(v) = L_{h} (v')$ and $v$ is \emph{older than} $v'$, then remove $q$ from $v'$ and all its descendants.
\item \label{en:STS-v-merge} \emph{Vertical Merge}. For each $v$, if all children of $v$ have index label $0$, then remove all descendants of $v$.
\end{enumerate}
\item \label{en:STS-accepance} $[\tilde{G}, \tilde{B}]_{\tilde{I}}$ is such that for every $i \in \tilde{I}$
\begin{align*}
\tilde{G}(i)&=\{ \tilde{q} \in \tilde{Q} \mid \tilde{q} \textit{ contains a red node with name $i$ or does not contains a node with name $i$} \, \} \\
\tilde{B}(i)&=\{ \tilde{q} \in \tilde{Q} \mid \tilde{q} \textit{ contains a green node with name $i$} \, \}
\end{align*}
\end{enumerate}
\end{procedure}

Step~\eqref{en:STS-expansion-nonleaf-Bs} says that if a run in $v_{j'}$ fulfills the positive obligation $i_{j'}$ by visiting $B(i_{j'})$ ($L_{s}(v_{j'}) \cap B(i_{j'}) \not = \emptyset$), then the run moves into a new node $v'$ ($L_{s}(v')=L_{s}(v_{j'}) \cap B(i_{j'})$), and $v'$ continues to monitor if the run hits the next largest positive obligation in the witness set of its parent ($L_{h}(v')=\max([0..i_{j'}) \cap ((I \cup \{0\}) \setminus L^{\set}_{h}(v)))$). Step~\eqref{en:STS-expansion-nonleaf-Gs} says that if this is not the case and the run in $v_{j'}$ also violates the negative obligation $i_{j'}$ by visiting $G(i_{j'})$ ($L_{s}(v_{j'}) \cap G(i_{j'}) \not = \emptyset$), then this run is \emph{reset}, in the sense that the states in $L_{s}(v_{j'}) \cap G(i_{j'})$ moved into a new child of $v$. 
\section{Improved Streett Determinization}
\label{sec:improved-Streett-determinization}

In this section, we show our improved construction for Streett determinization. Define $\mu=\min(n,k)$.

\subsection{Improvement I}
The first idea is what we have applied to Streett complementation~\cite{Cai+Zhang/11/Tight+Upper}, namely, the larger the $k$, the more overlaps between $B(i)$'s and between $G(i)$'s ($i \in I$). Let us revisit the previous example, illustrated in Figure~\ref{fig:witness-set-illustration}. Assume that $G(2) \subseteq G(3)$ (we say that $G(3)$ covers $G(2)$). If a run stays at $v_{3}$, then the run is not supposed to visit $G(3)$, or otherwise the run should have been reset by Step~\eqref{en:STS-expansion-nonleaf-Gs}. Since the run cannot visit $G(2)$ either, there is no point to check if it is to visit $B(2)$, and hence $v_{3}$ does not need to have a child with index label $2$ (in this case the node $v_{32}$). This simple idea already puts a cap on the size of $\STS$. But it turns out that we can save the most if we exploit the redundancy on $B$ instead of on $G$.

\paragraph{Reduction of Tree Size.}
Step~\eqref{en:STS-expansion-nonleaf-Bs} at a non-leaf node $v$ is to check, for every child $v'$ of $v$, if a run visits $B(L_{h}(v'))$, and in the positive case, move the run into a new node. Let $v'$ be a non-root node and $v$ the parent of $v'$. Let $I_{v}=L^{\set}_{h}(v)$. As Step~\eqref{en:STS-expansion} is executed recursively from top to bottom, it can be assured that at the moment of its arriving at $v$, for every $i \in I_{v}$, we have $L_{s}(v') \cap B(i) = \emptyset$. By an abuse of notation, we write $B(v)$ for $\cup_{j \in I_{v}} B(j)$, and hence we have $L_{s}(v') \cap B(v) = \emptyset$. Thus, there is no chance of missing a positive obligation even if we restrict $L_{h}$ to be such that $B(L_{h}(v')) \not \subseteq B(v)$. It follows that each node ``watches'' at least one more state that has not been watched by its ancestors, and therefore along any path of an $\STS$, there are at most $\mu$ nodes with non-zero index labels (recall that $\mu=\min(n,k)$ and note that the root is excluded as its index label is $0$). Also, it can be shown by induction on tree height that an $\STS$ contains at most $n$ nodes with index label $0$, using the fact that state labels of sibling nodes are pairwise disjoint and the fact that if $v$ is the parent of $v'$ and $L_{h}(v')=0$, then $L_{s}(v') \subset L_{s}(v)$. Therefore, the number of nodes in an $\STS$ is bounded by $n(\mu+1)$.

\paragraph{Reduction of Index Labels.}
Let $I'_{v}=\{i \in I \mid B(i) \subseteq B(v) \}$. The above analysis tells us that $L_{h}(v') \in (I \setminus I'_{v})$. However, we can further improve $L_{h}$ such that there is no $j \in (I \setminus I'_{v})$, $(B(j) \setminus B(v)) \subset (B(L_{h}(v')) \setminus B(v))$ and for any $j \in (I \setminus I'_{v})$, $(B(j) \setminus B(v)) = (B(L_{h}(v')) \setminus B(v))$ implies $L_{h}(v') < j$. We say that $L_{h}(v')$ minimally extends $L^{\to}_{h}(v)$ if this condition holds.

To formalize the intuition of \emph{minimal extension}, we introduce two functions $\Cover: I^{*} \to 2^{I}$ and $\Mini: I^{*} \to 2^{I}$ as in~\cite{Cai+Zhang/11/Tight+Upper}. $\Cover$ maps finite sequences of $I$-elements to subsets of $I$ such that
\begin{align*}
\Cover(\alpha) = \{\, j \in I \ \mid \ B(j) \subseteq \bigcup_{i=1}^{|\alpha|} B(\alpha[i]) \, \} .
\end{align*}
Note that $\Cover(\epsilon)=\emptyset$. $\Mini$ also maps finite sequences of $I$-elements to subsets of $I$ such that $j \in \Mini(\alpha)$ if and only if $j \in I \setminus \Cover(\alpha)$ and
\begin{align}
& \forall j' \in I \setminus \Cover(\alpha) \, \Big[ \, j' \not = j \ \to \
 \neg \,\Big(\, B(j') \cup \bigcup_{i=1}^{|\alpha|} B(\alpha[i]) \ \subseteq \ B(j) \cup \bigcup_{i=1}^{|\alpha|} B(\alpha[i]) \,\Big)\, \Big] \, , \label{eq:mini-1} \\
& \forall j' \in I \setminus \Cover(\alpha) \, \Big[ \, j' < j \ \to \
 \,\Big(\, B(j') \cup \bigcup_{i=1}^{|\alpha|} B(\alpha[i]) \ \not = \ B(j) \cup \bigcup_{i=1}^{|\alpha|} B(\alpha[i]) \,\Big)\, \Big] \, .\label{eq:mini-2}
\end{align}
$\Mini(\alpha)$ consists of index candidates to \emph{minimally} enlarge $\Cover(\alpha)$; ties (with respect to set inclusion) are broken by numeric minimality (Condition~\eqref{eq:mini-2}).

$\Mini$ plays a crucial role in reducing the combination of index labels. In the reduced Safra's trees (Definition~\ref{def:RSTS}), we identify a collection of paths such that each node appears on exactly one of those paths. The sequence of index labels on each of those paths corresponds to a path in a specific tree structure, which we refer to as \emph{increasing tree of sets} ($\ITS$)~\cite{Cai+Zhang/11/Tight+Upper}. Counting the number of paths in $\ITS$ give us a better upper bound on the combination of index labels. Here we switch to an informal notation of labeled trees and we identify a node with the sequence of labels from the root to the node.

\begin{definition}[Increasing Tree of Sets ($\ITS$)~\cite{Cai+Zhang/11/Tight+Upper}]
\label{def:ITS} An $\ITS$ $\calT(n,k,B)$ is an unordered $I$-labeled tree such that a node $\alpha$ exists in $\calT(n,k,B)$ if and only if $\forall i \in [1..|\alpha|]$, $\alpha[i] \in \Mini(\alpha[1..i))$.
\end{definition}

Several properties are easily seen from the definition. First, an $\ITS$ is uniquely determined by parameter $n$, $k$ and $B$. Second, the length of the longest path in $\calT(n,k,B)$ is bound by $\mu$. Third, if $\beta$ is a direct child of $\alpha$, then $\beta$ must contribute at least one new element that has not been seen from the root to $\alpha$. Forth, the new contributions made by $\beta$ cannot be covered by contributions made by another sibling $\beta'$, with ties broken by selecting the one with smallest index. As $B:I \to 2^{Q}$ is one to one, we also view $\ITS$ as $2^{Q}$-labeled trees.

\begin{figure}[t!]
\begin{center}
\begin{tikzpicture}[scale=0.65]
\Tree [.$0:\emptyset$
        [.$2:\{q_{0}\}$
          [.$1:\{q_{0},q_{1}\}$
            [.$3:\{q_{1},q_{2}\}$ ]
          ]
          [.$4:\{q_{2}\}$
            [.$1:\{q_{0},q_{1}\}$ ]
          ]
        ]
        [.$4:\{q_{2}\}$
          [.$3:\{q_{1},q_{2}\}$
            [.$2:\{q_{0}\}$ ]
          ]
          [.$2:\{q_{0}\}$
            [.$1:\{q_{0},q_{1}\}$ ]
          ]
        ]
      ]
\end{tikzpicture}
\end{center}
\caption{\textrm{The $\ITS$ $\calT(n,k,B)$ in Example~\ref{ex:ITS}. Note that $\{q_{0},q_{1}\}$ and $\{q_{1},q_{2}\}$ cannot appear in the first level because $\{q_{0},q_{1}\}$ covers $\{q_{0}\}$ and $\{q_{1},q_{2}\}$ covers $\{q_{2}\}$. The leftmost node at the bottom level is labeled by $\{q_{1},q_{2}\}$ instead of by $\{q_{2}\}$ due to the index minimality requirement}.}
\label{fig:ITS}
\end{figure}

\begin{example}[$\ITS$]
\label{ex:ITS}
Consider $n=3$, $k=4$, $Q=\{q_{0},q_{1},q_{2}\}$, and $B:[1..4] \to 2^{Q}$ such that
\begin{align*}
B(1)=\{q_{0}, q_{1}\}, && B(2)=\{q_{0}\}, && B(3)=\{q_{1},q_{2}\}, && B(4)=\{q_{2}\} \,.
\end{align*}
Figure~\ref{fig:ITS} shows the corresponding $\calT(n,k,B)$.
\end{example}

\subsection{Improvement II}

The second idea is a batch-mode naming scheme to reduce name combinations. As shown before, an $\STS$ can have $n(\mu+1)$ nodes, which translates to $(n(\mu+1))!$ name combinations according to the current ``first-come-first-serve'' naming scheme, that is, picking an unused name when a new node is created and recycling the name when a node is removed. When $k=O(n)$, the naming cost is higher than all other complexity factors combined, as $(n(\mu+1))!=2^{O(nk \lg nk)}$. However, this can be overcome by dividing the name space into even buckets and ``wholesaling'' buckets to specific paths in an $\STS$.

\paragraph{Reduction of Names.}
Let $t$ be an $\STS$. A \emph{left spine} ($\LS$) is a maximal path $l=v_{1} \cdots v_{m}$ such that $v_{m}$ is a leaf, for any $i \in [2..m]$, $v_{i}$ is the left-most child of $v_{i-1}$, and $v_{1}$ is not a left-most child of its parent. We call $v_{1}$ the head of $l$. Let $\head: V \to V$ be such that $\head(v)=v'$ if $v'$ is the head of the $\LS$ where $v$ belongs to. We say that $l$ is the $i$-th $\LS$ of $t$ if $v_{m}$ is the $i$-th leaf of $t$, counting from left to right. It is clear that a tree with $m$ leaves has $m$ $\LS$, and every node is on exactly one $\LS$. Thus, an $\STS$ has at most $n$ $\LS$. For this renaming scheme to work, we require that a newly created node be added to the right of all existing siblings of the same $L_{h}$-value. If a non-head node on an $\LS$ has index label $0$, then so should all of its siblings. But then all of them should have been removed. Therefore, only the head of an $\LS$ can have index label $0$, which means that an $\LS$ can have at most $\mu+1$ nodes.

We use $n(\mu+1)$ names and divide them evenly into $n$ buckets. Let $b_{i}$ ($i \in [1..n]$) denote the $i$-th bucket, $b_{ij}=(\mu+1)(i-1)+j$ ($i \in [1..n], j \in [1..\mu+1]$) the $j$-th name in the $i$-th bucket. We say that $b_{i1}=(\mu+1)(i-1)+1$ is the \emph{initial value} of $b_{i}$. Our naming strategy is as follows. Every $\LS$ $l=v_{1}\cdots v_{m}$ in an $\STS$ is associated with a name bucket $b$ and $v_{1}, \cdots, v_{m}$ are assigned names continuously from the initial value of $b$. For example, if $l$ is associated with bucket $b_{t}$, then $L_{n}(v_{i})=(\mu+1)(t-1)+i$ for $i \in [1..m]$. The bucket association for each $\LS$ in an $\STS$ $t$ can be viewed as selection function $\bucket: V \to [1..n]$ such that $\bucket(v)=i$ if node $v$ is assigned a name in the $i$-th bucket.

This naming strategy, however, comes with a complication; what if a leftmost sibling $v$ is removed (due to Step~\eqref{en:STS-expansion-nonleaf}), and the second leftmost sibling $v'$ (if exists) and all nodes belonging to the $\LS$ of which $v'$ is the head, ``graft into'' the $\LS$ that $v$ belongs to? The answer is that at the end of tree transformation, we need to rename those nodes that have moved into another $\LS$. If a tree transformation turns $t$ into $t'$, and during the process a node $v$ joins another $\LS$ $l$ in $t'$ (which is also in $t$ before the transformation), then we rename $v$ to a name in the bucket that $l$ uses in $t$, and recycle the bucket with which $v$ was associated in $t$.

\begin{example}[New Naming Scheme]
\label{ex:RSTS-naming}
Figure~\ref{fig:STS-naming} illustrates the changes of names in a sequence of tree transformations. We assume that there are $10$ buckets $b_{1}-b_{10}$, each of which is of size $4$. Nodes in the graphs are denoted in the form $v:L_{n}(v)$; all other types of labels are omitted for simplicity.
\end{example}

\begin{figure}[ht!]%
\centering
\subfloat[][]{%
\label{fig:STS-naming-a}%
\begin{tikzpicture}[scale=0.55]
\Tree [.$v_{0}:1$
        [.$v_{1}:2$
          [.$v_{2}:3$
            [.$v_{3}:4$ ]
            [.$v_{4}:17$ ]
          ]
          [.$v_{5}:25$
            [.$v_{6}:26$ ]
            [.$v_{7}:5$ ]
          ]
        ]
        [.$v_{8}:21$
          [.$v_{9}:22$ ]
          [.$v_{10}:33$ ]
        ]
        [.$v_{11}:13$ ]
      ]
\end{tikzpicture}
}%
\hspace{3pt}%
\subfloat[][]{%
\label{fig:STS-naming-b}%
\begin{tikzpicture}[scale=0.55]
\Tree [.$v_{0}:1$
        [.$v_{1}:2$
          [.$v_{5}:3$
            [.$v_{6}:4$ ]
            [.$v_{7}:5$ ]
          ]
        ]
        [.$v_{8}:21$
          [.$v_{9}:22$ ]
          [.$v_{10}:33$ ]
        ]
        [.$v_{11}:13$ ]
      ]
\end{tikzpicture}
} \\
\subfloat[][]{%
\label{fig:STS-naming-c}%
\begin{tikzpicture}[scale=0.55]
\Tree [.$v_{0}:1$
        [.$v_{1}:2$
          [.$v_{5}:3$
            [.$v_{6}:4$ ]
            [.$v_{7}:5$ ]
          ]
          [.$v_{12}:25$
            [.$v_{13}:26$ ]
          ]
        ]
        [.$v_{8}:21$
          [.$v_{9}:22$ ]
          [.$v_{10}:33$ ]
          [.$v_{14}:29$ ]
        ]
        [.$v_{11}:13$ ]
      ]
\end{tikzpicture}
}%
\hspace{3pt}%
\subfloat[][]{%
\label{fig:STS-naming-d}%
\begin{tikzpicture}[scale=0.55]
\Tree [.$v_{0}:1$
        [.$v_{8}:2$
          [.$v_{9}:3$ ]
          [.$v_{10}:33$ ]
          [.$v_{14}:29$ ]
        ]
        [.$v_{11}:13$ ]
      ]
\end{tikzpicture}
}
\caption[Renaming of an $\STS$]{
\subref{fig:STS-naming-a} shows an $\STS$ with $7$ $\LS$: $l_{1}:v_{0}v_{1}v_{2}v_{3}$, $l_{2}:v_{4}$, $l_{3}:v_{5}v_{6}$, $l_{4}:v_{7}$, $l_{5}:v_{8}v_{9}$, $l_{6}:v_{10}$ and $l_{7}:v_{11}$, associated with buckets $b_{1}$, $b_{5}$, $b_{7}$, $b_{2}$, $b_{6}$, $b_{9}$ and $b_{4}$, respectively.
\subref{fig:STS-naming-b} shows the resulting $\STS$ after removing $v_{2}$ and its descendants. Nodes $v_{5}$ and $v_{6}$ migrate into $l_{1}$, and accordingly their name bucket $b_{7}$ is recycled. Also recycled is bucket $b_{5}$ due to the deletion of $v_{4}$.
\subref{fig:STS-naming-c} shows the resulting $\STS$ after adding nodes $v_{12}$, $v_{13}$ and $v_{14}$ and forming two new $\LS$: $v_{12}v_{13}$ and $v_{14}$. Here $v_{12}$ and $v_{13}$ reuse the previously recycled bucket $b_{7}$, while node $v_{14}$ takes an unused bucket $b_{8}$.
\subref{fig:STS-naming-d} shows the resulting $\STS$ after removing $v_{1}$ and its descendants. Nodes $v_{8}$ and $v_{9}$ migrate into $l_{1}$ and take names $2$ and $3$, respectively. Buckets $b_{2}$, $b_{7}$ and $b_{6}$ are recycled accordingly.}
\label{fig:STS-naming}%
\end{figure}

\begin{definition}[Reduced Safra Trees for Streett Determinization ($\RSTS$)]
\label{def:RSTS}
A \emph{reduced Safra tree for Streett determinization} ($\RSTS$) is an $\STS$ $\langle V, L \rangle$ with $L=\langle L_{n}, L_{s}, L_{c}, L_{h} \rangle$ that satisfies the following additional conditions:
\begin{enumerate}[label=\ref{def:RSTS}.\arabic*,ref=\ref{def:RSTS}.\arabic*]
\item\label{en:RSTS-Lh-extra} \emph{Condition on $L_{h}$}. For each node $v$, if $\Mini(L^{\to}_{h}(v))\not = \emptyset$, then $v$ is not a leaf node and for any child $v'$ of $v$, $L_{h}(v') \in \Mini(L^{\to}_{h}(v))$.
\item\label{en:RSTS-Ln-extra} \emph{Condition on $L_{n}$}. There exists a function $\bucket: V \to [1..n]$ such that for every $\LS$ $v_{1} \cdots v_{m}$, we have $\bucket(v_{i})=\bucket(v_{j})$ for $i,j \in [1..m]$ and $L_{n}(v_{i}) = (\mu+1)(\bucket(v_{i})-1)+i$ for $i \in [1..m]$.
\end{enumerate}
\end{definition}
As before, nodes in a $\RSTS$ are partially ordered by \emph{older-than} relation. But we impose an additional \emph{structural ordering} on nodes, that, for any two sibling $v$ and $v'$, $v'$ is placed to the right of $v$ if and only if $L_{h}(v) > L_{h} (v')$, or $L_{h}(v) = L_{h} (v')$ and $v$ is \emph{older than} $v'$. This structural ordering is needed for our renaming scheme (see the proof of Theorem~\ref{thm:Streett-determinization-complexity}).

Condition~\eqref{en:RSTS-Ln-extra}, together with the requirement that $L_{n}$ is injective (Condition~\eqref{en:STS-name}), guarantees that no two distinct $\LS$ in a tree share a bucket. Condition~\eqref{en:RSTS-Lh-extra} says that every $\RSTS$ has fully grown left spines, that is, no leaf $v$ can be further extended, as $\Mini(L^{\to}_{h}(v))=\emptyset$. To achieve this, we need the following procedure applied as the last step of each tree transformation.

\begin{procedure}[Grow Left Spine]
\label{pro:grow-left-spine}
Repeat the following procedure until no new nodes can be added: if $v$ is a leaf and $\Mini(L^{\to}_{h}(v)) \not=\emptyset$, add a new child $v'$ to $v$ with $L_{s}(v')=L_{s}(v)$, $L_{h}(v')=\max(\Mini(L^{\to}_{h}(v)))$, $L_{c}(v')=\red$, and $L_{n}(v')=L_{n}(v)+1$.
\end{procedure}
We note that the requirement that a $\RSTS$ has fully grown left spines is not essential; we can ``grow'' a $\RSTS$ ``on-the-fly'' as in~\cite{Safra/92/Exponential,Piterman/06/From}. But this requirement simplifies the analysis on the number of combinations of index labels (see the proof of Theorem~\ref{thm:Streett-determinization-complexity}).

\begin{Rule}[Naming and Coloring on $\RSTS$]
\label{rule:RSTS-naming}
Naming and Coloring convention for $\RSTS$ is the one for $\STS$ plus the following: (1) nodes in an $\LS$, from the head downwards, are assigned continuously increasing names, starting from the initial value of a bucket; (2) when an $\LS$ is created, nodes in the $\LS$ are assigned names from an unused name bucket; when an $\LS$ is removed (which only happens when its head is removed), the name bucket of the $\LS$ is recycled; (3) when an $\LS$ $l$ is grafted into another $\LS$ $l'$, the name bucket of $l$ is recycled and nodes on $l$ are renamed according to (1), as if they were on $l'$ originally; (4) renamed nodes are marked $\red$.
\end{Rule}

\begin{procedure}[Improved Streett Determinization]
\label{pro:improved-Streett-determinization}
Let $\calA=\langle Q, Q_{0}, \Sigma, \Delta, \langle G, B \rangle_{I} \rangle$ be a nondeterministic Streett automaton. This procedure outputs a deterministic Rabin automaton $\calB=\langle \tilde{Q}, \tilde{q}_{0}, \Sigma, \tilde{\Delta}, [\tilde{G}, \tilde{B}]_{\tilde{I}} \rangle$, where
$\tilde{Q}$ is the set of $\RSTS$, $\tilde{I}=[1..n(\mu+1)]$, $[\tilde{G}, \tilde{B}]_{\tilde{I}}$ is as defined in Procedure~\ref{pro:Safra-Streett-determinization}, $\tilde{q}_{0}$ is a tree that is just a fully grown left spine obtained by growing the single root tree (as defined in Procedure~\ref{pro:Safra-Streett-determinization}) according to Procedure~\eqref{pro:grow-left-spine}, and $\tilde{\Delta}$ is defined such that $\tilde{q}'=\tilde{\Delta}(\tilde{q})$ if and only if $\tilde{q}'$ is the $\RSTS$ obtained by applying the following transformation rule to $\tilde{q}$.
\begin{enumerate}[label=\ref{pro:improved-Streett-determinization}.\arabic*,ref=\ref{pro:improved-Streett-determinization}.\arabic*]
\item \label{en:RSTS-subset-construction} \emph{Subset Construction}. For each node $v$ in $\tilde{Q}$, update state label $L_{s}(v)$ to $\Delta(L_{s}(v))$.
\item \label{en:RSTS-expansion} \emph{Expansion}. Apply the following transformations to non-leaf nodes recursively from the root. Let $v$ be a node with $j$ children $v_{1}, \ldots, v_{j}$ with $i_{1}, \ldots, i_{j}$ as the corresponding index labels. Consider the following cases for each $j' \in [1..j]$.
\begin{enumerate}[label=\ref{en:RSTS-expansion}.\arabic*,ref=\ref{en:RSTS-expansion}.\arabic*]
\item \label{en:RSTS-expansion-Bs} If $L_{s}(v_{j'}) \cap B(i_{j'}) \not = \emptyset$, then add a child $v'$ to $v$ with $L_{s}(v')=L_{s}(v_{j'}) \cap B(i_{j'})$ and $L_{h}(v')=\max([0..i_{j'}) \cap (\{0\} \cup \Mini(L^{\to}_{h}(v))))$, and remove the states in $L_{s}(v_{j'}) \cap B(i_{j'})$ from $v_{j'}$ and all the descendants of $v_{j'}$.
\item \label{en:RSTS-expansion-Gs} If $L_{s}(v_{j'}) \cap B(i_{j'}) = \emptyset$ and $L_{s}(v_{j'}) \cap G(i_{j'}) \not = \emptyset$, then add a child $v'$ to $v$ with $L_{s}(v')=L_{s}(v_{j'}) \cap G(i_{j'})$ and $L_{h}(v')=L_{h}(v)$, and remove the states in $L_{s}(v_{j'}) \cap G(i_{j'})$ from $v_{j'}$ and all the descendants of $v_{j'}$.
\end{enumerate}
\item \label{en:RSTS-h-merge} \emph{Horizontal Merge}. For any state $q$ and any two siblings $v$ and $v'$ such that $q \in L_{s}(v) \cap L_{s}(v')$, if $L_{h}(v) < L_{h} (v')$, or $L_{h}(v) = L_{h} (v')$ and $v$ is \emph{older than} $v'$, then remove $q$ from $v'$ and all its descendants.
\item \label{en:RSTS-v-merge} \emph{Vertical Merge}. For each $v$, if all children of $v$ have index label $0$, then remove all descendants of $v$.
\item \label{en:RSTS-grow-LS} Grow the tree fully according to Procedure~\eqref{pro:grow-left-spine}.
\end{enumerate}
\end{procedure}
Note that nodes are assigned or reassigned names in batch only after the whole expansion phase is carried out and the tree is fully grown. Besides naming and renaming, the only major difference between Procedures~\ref{pro:improved-Streett-determinization} and~\ref{pro:Safra-Streett-determinization} is the way of selecting the next positive obligation. In Step~\eqref{en:RSTS-expansion-Bs}, $\Mini(L^{\to}_{h}(v))$ is used in the calculation, while in Step~\eqref{en:STS-expansion-nonleaf-Bs}, $L^{\set}_{h}(v)$ is used.

\begin{theorem}[Streett Determinization: Correctness]
\label{thm:Streett-determinization-correctness}
Let $\calA=\langle Q, Q_{0}, \Sigma, \Delta, \langle G, B \rangle_{I} \rangle$ be a Streett automaton with $|Q|=n$ and $I=[1..k]$, and $\calB=\langle \tilde{Q}, \tilde{q}_{0}, \Sigma, \tilde{\Delta}, [\tilde{G}, \tilde{B}]_{\tilde{I}} \rangle$ the deterministic Rabin automaton obtained by Procedure~\ref{pro:improved-Streett-determinization}. We have $\scrL(\calA)=\scrL(\calB)$.
\end{theorem}

\begin{proof}
($\scrL(\calA) \subseteq \scrL(\calB)$). This part of proof is almost identical to the one in~\cite{Safra/92/Exponential}. We ought to show that if $\varrho=\tilde{q}_{0} \tilde{q}_{1} \cdots$ is an run of $\calB$ over an infinite word $w=w_{0}w_{1} \ldots \in \scrL(\calA)$, then (1) a node $v$ exists in every state in $\varrho$ from some point on, (2) $v$ turns green infinitely often, and (3) $v$ has a fixed name $i \in \tilde{I}$. The argument in~\cite{Safra/92/Exponential} guarantees the existence of such a node $v$ with the first two properties. The only complication comes from renaming. We have the situation that $v$ with name $i$ exists in $\tilde{q}_{j}$, but it is renamed to $i'$ in the following state $\tilde{q}_{j+1}$. This happens when $v$ is on an $\LS$ $l$ whose head is a second left-most sibling in $\tilde{q}_{j}$ and the corresponding leftmost sibling is removed in $\tilde{q}_{j+1}$, resulting in nodes on $l$ (including $v$) joining another $\LS$ $l'$ in $\tilde{q}_{j+1}$. However, such ``grafting'' can only happen to $v$ finitely many times, as each time the height of the head of $l'$ is strictly smaller than the height of the head of $l$. Therefore, $v$ is eventually assigned a fixed name $i$, which gives us the third property.

($\scrL(\calB) \subseteq \scrL(\calA)$). We ought to show that if $w=w_{0}w_{1} \ldots \in \scrL(\calB)$, then the run $\varrho=\tilde{q}_{0} \tilde{q}_{1} \cdots$ of $\calB$ over $w$ induces an accepting run of $\calA$ over $w$. The assumption that $\varrho$ is accepting means that there exists an $i \in \tilde{I}$ such that $\varrho$ eventually never visits $\tilde{G}(i)$, but visits $\tilde{B}(i)$ infinitely often, or equivalently, $i$ names a green or yellow node in every state in a suffix of $\varrho$ and there are infinitely many occurrences when the nodes named by $i$ are green. Since renamed nodes are marked red, all nodes named by $i$ in the suffix have to be the same one. It follows that a node $v$ eventually stays in every state in a suffix of $\varrho$ and $v$ turns green infinitely often. The rest of the proof is the same as the one in~\cite{Safra/92/Exponential}, with the help of Lemma~\ref{lem:Streett-connecting-dots} and the fact that using $\Mini$ to select the index labels of the children of $v$ is sound and complete.
\end{proof}
\section{Complexity}
\label{sec:complexity}

In this section we state the complexity results for the determinization of Streett, generalized B\"{u}chi, parity and Rabin automata.
\begin{theorem}[Streett Determinization: Complexity]
\label{thm:Streett-determinization-complexity}
Let $\calA=\langle Q, Q_{0}, \Sigma, \Delta, \langle G, B \rangle_{I} \rangle$ be a Streett automaton with $|Q|=n$ and $I=[1..k]$, and $\calB=\langle \tilde{Q}, \tilde{q}_{0}, \Sigma, \tilde{\Delta}, [\tilde{G}, \tilde{B}]_{\tilde{I}} \rangle$ the deterministic Rabin automaton obtained by Procedure~\ref{pro:improved-Streett-determinization}. For any $k$, $|\tilde{Q}| \le 2^{12n}(0.37n)^{n^{2}+8n}=2^{O(n^{2} \lg n)}$ and $|\tilde{I}|=O(n^{2})$. If $k = O(n)$, we have $|\tilde{Q}|=2^{O(n \lg n + nk \lg k)}$ and $|\tilde{I}|=O(nk)$.
\end{theorem}

\let\oldthefootnote\thefootnote
\renewcommand{\thefootnote}{\fnsymbol{footnote}}

\begin{figure}
\begin{center}
\begin{tabular}{|l|l|l|ll|}
\hline
Cost             & Safra's & Piterman's & Ours    & \\ \hline
Ordered Tree       & $2^{O(nk)}$ &   $2^{O(nk)}$      & $2^{O(n \lg n)}$ & \\ \hline
Name  &    $2^{O(nk \lg nk)}$  &  $2^{O(nk \lg nk)}$\footnotemark[2] & $2^{O(n \lg n)}$ &  \\ \hline
\multirow{2}{*}{Color} & \multirow{2}{*}{$2^{O(nk)}$} & \multirow{2}{*}{$O((nk)^{2})$\footnotemark[2]} & $2^{O(n \lg k)}$ & $k = O(n)$ \\ \cline{4-5}
                 &                      & & $2^{O(n \lg n)}$ & $k=\omega(n)$ \\ \hline
Set Label & $2^{O(n \lg n)}$  & $2^{O(n \lg n)}$  & $2^{O(n \lg n)}$ &  \\ \hline
\multirow{2}{*}{Index Label} & \multirow{2}{*}{$2^{O(nk \lg nk)}$} & \multirow{2}{*}{$2^{O(nk \lg nk)}$} & $2^{O(n \lg n + nk \lg k)}$ & $k = O(n)$ \\ \cline{4-5}
                 &                     &  & $2^{O(n^{2} \lg n)}$ & $k=\omega(n)$ \\ \hline
\multirow{2}{*}{Total} & \multirow{2}{*}{$2^{O(nk \lg nk)}$} & \multirow{2}{*}{$2^{O(nk \lg nk)}$} & $2^{O(n \lg n + nk \lg k)}$ & $k = O(n)$ \\ \cline{4-5}
                 &                     &  & $2^{O(n^{2} \lg n)}$ & $k=\omega(n)$ \\ \hline
\end{tabular}
\end{center}
\caption{Cost breakdown of Streett determinization.}
\label{fig:cost-break-down}
\end{figure}
\footnotetext[2]{In Piterman's construction, the number of ordered trees times the number of name combinations is bounded by $(nk)^{nk}$. However, the second factor is still the dominating one, costing $2^{\Omega(nk \lg nk)}$. Also, there is no notion of color in Piterman's construction. Instead, each tree is associated with two special names both ranging from $1$ to $nk$, resulting in the cost $O((nk)^{2})$.}

\let\thefootnote\oldthefootnote

Figure~\ref{fig:cost-break-down} breaks down the total cost into five categories and compares our construction with previous ones in each category. It turns out that the complexity analysis can be easily adapted for generalized B\"{u}chi and parity automata as they are subclasses of Streett automata.

\begin{corollary}[Generalized B\"{u}chi Determinization: Complexity]
\label{cor:GB-upper-bound}
Let $\calA$ be a generalized B\"{u}chi automaton with state size $n$ and index size $k$. There is an equivalent deterministic Rabin automaton $\calB$ with state size $2^{O(n \lg nk)}$. The index size is $O(nk)$ if $k=O(n)$ and $O(n^{2})$ if $k = \omega(n)$.
\end{corollary}

\begin{corollary}[Parity Determinization: Complexity]
\label{cor:parity-upper-bound}
Let $\calA$ be a parity automaton with state size $n$ and index size $k$. There is an equivalent deterministic Rabin automaton $\calB$ with state size $2^{O(n \lg n)}$ and index size $O(nk)$.
\end{corollary}

The determinization of a Rabin automaton with acceptance condition $[G, B]_{I}$ (the dual of $\langle G,B\rangle_{I}$ for $I=[1..k]$) can be straightforwardly obtained by running, in parallel, $k$ modified Safra trees, each of which monitors runs for an individual Rabin condition $[G(i), B(i)]$ ($i \in I$).

\begin{theorem}[Rabin Determinization: Complexity]
\label{the:rabin-upper-bound}
Let $\calA$ be a Rabin automaton with state size $n$ and index size $k$. There is an equivalent deterministic Rabin automaton $\calB$ with state size $2^{O(nk \lg n)}$ and index size $O(nk)$.
\end{theorem}

We note that it is unlikely that there exists a Safra-tree style determinization for Rabin automata, because an analogue of Lemma~\ref{lem:Streett-connecting-dots} fails due to the existential nature of Rabin acceptance conditions. 
\section{Concluding Remarks}
\label{sec:conclusion}
In this paper we improved Safra's construction and obtained tight upper bounds on the determinization complexities of Streett, generalized B\"{u}chi, parity and Rabin automata. Figure~\ref{fig:bound-summary} summarizes these complexity results.

Our results show an interesting phenomenon that in the asymptotic notation $2^{\Theta(\cdot)}$, complementation complexity is identical to determinization complexity. The same phenomenon happens to finite automata on finite words. We believe it is worth investigating the reason behind this phenomenon.

As mentioned earlier, determinization procedures that output parity automata, like Piterman's constructions, are preferable to the classic ones that output Rabin automata. We plan to investigate how to combine Piterman's node renaming scheme with ours to obtain determinization procedures that output parity automata with optimal state complexity.

\begin{figure}[t!]
\begin{center}
\begin{tabular}{|l||ll|l|l|}
\hline
Type        & Bound                  &                       & Lower            & Upper        \\ \hline
B\"{u}chi        & $2^{\Theta(n \lg n)}$  &                  & \cite{Michel/88/Complementation}     & \cite{Safra/88/Complexity} \\ \hline
Generalized B\"{u}chi       & $2^{\Theta(n \lg nk)}$ &   & \cite{Yan/06/Lower}     & \cite{Kupferman+Vardi/05/Complementation} \\ \hline
\multirow{2}{*}{Streett}  & $2^{\Theta(n \lg n + nk \lg k)}$ & $k = O(n)$      & \multirow{2}{*}{\cite{Cai+Zhang/11/Tight+Lower}}  & \multirow{2}{*}{\cite{Cai+Zhang/11/Tight+Upper}} \\
                             & $2^{\Theta(n^{2} \lg n)}$ & $k=\omega(n)$   &      &       \\ \hline
Rabin    & $2^{\Theta(nk \lg n)}$ &    & \cite{Cai+Zhang+Luo/09/Improved}     & \cite{Kupferman+Vardi/05/Complementation} \\ \hline
Parity   & $2^{\Theta(n \lg n)}$  &    & \cite{Michel/88/Complementation}     & \cite{Cai+Zhang/11/Tight+Upper}         \\ \hline
\end{tabular}
\end{center}
\caption{\textrm{Complementation and determinization complexities for $\omega$-automata of common types. The listed citations are meant for complementation only.}}
\label{fig:bound-summary}
\end{figure} 

\bibliographystyle{eptcs}
\bibliography{automata}

\begin{thebibliography}{10}
\providecommand{\bibitemdeclare}[2]{}
\providecommand{\surnamestart}{}
\providecommand{\surnameend}{}
\providecommand{\urlprefix}{Available at }
\providecommand{\url}[1]{\texttt{#1}}
\providecommand{\href}[2]{\texttt{#2}}
\providecommand{\urlalt}[2]{\href{#1}{#2}}
\providecommand{\doi}[1]{doi:\urlalt{http://dx.doi.org/#1}{#1}}
\providecommand{\bibinfo}[2]{#2}

\bibitemdeclare{incollection}{Buchi/66/Decision}
\bibitem{Buchi/66/Decision}
\bibinfo{author}{Richard \surnamestart B{\"{u}}chi\surnameend}
  (\bibinfo{year}{1966}): \emph{\bibinfo{title}{Symposium on Decision Problems:
  On a Decision Method in Restricted Second Order Arithmetic}}.
\newblock In: {\sl \bibinfo{booktitle}{Logic, Methodology and Philosophy of
  Science, Proceeding of the 1960 International Congress}}, {\sl
  \bibinfo{series}{Studies in Logic and the Foundations of
  Mathematics}}~\bibinfo{volume}{44}, \bibinfo{publisher}{Elsevier}, pp.
  \bibinfo{pages}{1--11}, \doi{10.1016/S0049-237X(09)70564-6}.

\bibitemdeclare{inproceedings}{Cai+Zhang/11/Tight+Lower}
\bibitem{Cai+Zhang/11/Tight+Lower}
\bibinfo{author}{Yang \surnamestart Cai\surnameend} \& \bibinfo{author}{Ting
  \surnamestart Zhang\surnameend} (\bibinfo{year}{2011}):
  \emph{\bibinfo{title}{A Tight Lower Bound for Streett Complementation}}.
\newblock In: {\sl \bibinfo{booktitle}{Proceedings of the IARCS Annual
  Conference on Foundations of Software Technology and Theoretical Computer
  Science (FSTTCS 2011)}}, pp. \bibinfo{pages}{339--350},
  \doi{10.4230/LIPIcs.FSTTCS.2011.339}.

\bibitemdeclare{inproceedings}{Cai+Zhang/11/Tight+Upper}
\bibitem{Cai+Zhang/11/Tight+Upper}
\bibinfo{author}{Yang \surnamestart Cai\surnameend} \& \bibinfo{author}{Ting
  \surnamestart Zhang\surnameend} (\bibinfo{year}{2011}):
  \emph{\bibinfo{title}{Tight Upper Bounds for Streett and Parity
  Complementation}}.
\newblock In: {\sl \bibinfo{booktitle}{Proceedings of the 20th Conference on
  Computer Science Logic (CSL 2011)}}, \bibinfo{publisher}{Dagstuhl
  Publishing}, pp. \bibinfo{pages}{112--128},
  \doi{10.4230/LIPIcs.CSL.2011.112}.

\bibitemdeclare{inproceedings}{Cai+Zhang+Luo/09/Improved}
\bibitem{Cai+Zhang+Luo/09/Improved}
\bibinfo{author}{Yang \surnamestart Cai\surnameend}, \bibinfo{author}{Ting
  \surnamestart Zhang\surnameend} \& \bibinfo{author}{Haifeng \surnamestart
  Luo\surnameend} (\bibinfo{year}{2009}): \emph{\bibinfo{title}{An Improved
  Lower Bound for the Complementation of Rabin Automata}}.
\newblock In: {\sl \bibinfo{booktitle}{Proceedings of the 24th Annual IEEE
  Symposium on Logic In Computer Science}}, \bibinfo{publisher}{IEEE Computer
  Society}, pp. \bibinfo{pages}{167--176}, \doi{10.1109/LICS.2009.13}.

\bibitemdeclare{article}{Choueka/74/Theories}
\bibitem{Choueka/74/Theories}
\bibinfo{author}{Yaacov \surnamestart Choueka\surnameend}
  (\bibinfo{year}{1974}): \emph{\bibinfo{title}{Theories of automata on
  $\omega$-types: a simplified appraoch}}.
\newblock {\sl \bibinfo{journal}{Journal of Computer and System Science}}
  \bibinfo{volume}{8}(\bibinfo{number}{2}), pp. \bibinfo{pages}{117--141},
  \doi{10.1016/S0022-0000(74)80051-6}.

\bibitemdeclare{inproceedings}{Colcombet+Zdanowski/09/Tight}
\bibitem{Colcombet+Zdanowski/09/Tight}
\bibinfo{author}{Thomas \surnamestart Colcombet\surnameend} \&
  \bibinfo{author}{Konrad \surnamestart Zdanowski\surnameend}
  (\bibinfo{year}{2009}): \emph{\bibinfo{title}{A Tight Lower Bound for
  Determinization of Transition Labeled {B\"{u}chi} Automata}}.
\newblock In: {\sl \bibinfo{booktitle}{Proceedings of the 36th Internatilonal
  Collogquium on Automata, Languages and Programming (ICALP 2009)}},
  \bibinfo{publisher}{Springer-Verlag}, pp. \bibinfo{pages}{151--162},
  \doi{10.1007/978-3-642-02930-1\_13}.

\bibitemdeclare{book}{Francez/1986/Fairness}
\bibitem{Francez/1986/Fairness}
\bibinfo{author}{Nissim \surnamestart Francez\surnameend}
  (\bibinfo{year}{1986}): \emph{\bibinfo{title}{Fairness}}.
\newblock \bibinfo{publisher}{Springer-Verlag New York, Inc.}

\bibitemdeclare{inproceedings}{Francez/84/Generalized}
\bibitem{Francez/84/Generalized}
\bibinfo{author}{Nissim \surnamestart Francez\surnameend} \&
  \bibinfo{author}{Dexter \surnamestart Kozen\surnameend}
  (\bibinfo{year}{1984}): \emph{\bibinfo{title}{Generalized fair termination}}.
\newblock In: {\sl \bibinfo{booktitle}{Proceedings of the 11th ACM
  SIGACT-SIGPLAN symposium on Principles of programming languages (POPL'84)}},
  \bibinfo{publisher}{ACM}, pp. \bibinfo{pages}{46--53},
  \doi{10.1145/800017.800515}.

\bibitemdeclare{inproceedings}{Gurevich+Harrington/82/Trees}
\bibitem{Gurevich+Harrington/82/Trees}
\bibinfo{author}{Yuri \surnamestart Gurevich\surnameend} \&
  \bibinfo{author}{Leo \surnamestart Harrington\surnameend}
  (\bibinfo{year}{1982}): \emph{\bibinfo{title}{Trees, Automata, and Games}}.
\newblock In: {\sl \bibinfo{booktitle}{Proceedings of the 14th annual ACM
  symposium on Theory of computing (STOC'82)}}, pp. \bibinfo{pages}{60--65},
  \doi{10.1145/800070.802177}.

\bibitemdeclare{incollection}{Kupferman+Vardi/05/Complementation}
\bibitem{Kupferman+Vardi/05/Complementation}
\bibinfo{author}{Orna \surnamestart Kupferman\surnameend} \&
  \bibinfo{author}{Moshe~Y. \surnamestart Vardi\surnameend}
  (\bibinfo{year}{2005}): \emph{\bibinfo{title}{Complementation Constructions
  for Nondeterministic Automata on Infinite Words}}.
\newblock In: {\sl \bibinfo{booktitle}{Tools and Algorithms for the
  Construction and Analysis of Systems}}, {\sl \bibinfo{series}{Lecture Notes
  in Computer Science}} \bibinfo{volume}{3440}, \bibinfo{publisher}{Springer
  Berlin / Heidelberg}, pp. \bibinfo{pages}{206--221},
  \doi{10.1007/978-3-540-31980-1\_14}.

\bibitemdeclare{inproceedings}{Loding/99/Optimal}
\bibitem{Loding/99/Optimal}
\bibinfo{author}{Christof \surnamestart L{\"{o}}ding\surnameend}
  (\bibinfo{year}{1999}): \emph{\bibinfo{title}{Optimal Bounds for
  Transformations of omega-Automata}}.
\newblock In: {\sl \bibinfo{booktitle}{Proceedings of the 19th Conference on
  Foundations of Software Technology and Theoretical Computer Science}},
  \bibinfo{publisher}{Springer-Verlag}, pp. \bibinfo{pages}{97--109},
  \doi{10.1007/3-540-46691-6\_8}.

\bibitemdeclare{article}{McNaughton/66/Testing}
\bibitem{McNaughton/66/Testing}
\bibinfo{author}{Robert \surnamestart McNaughton\surnameend}
  (\bibinfo{year}{1966}): \emph{\bibinfo{title}{Testing and Generating Infinite
  Sequences by a Finite Automaton}}.
\newblock {\sl \bibinfo{journal}{Information and Control}}
  \bibinfo{volume}{9}(\bibinfo{number}{5}), pp. \bibinfo{pages}{521--530},
  \doi{10.1016/S0019-9958(66)80013-X}.

\bibitemdeclare{misc}{Michel/88/Complementation}
\bibitem{Michel/88/Complementation}
\bibinfo{author}{M.~\surnamestart Michel\surnameend} (\bibinfo{year}{1988}):
  \emph{\bibinfo{title}{Complementation is more difficult with automata on
  infinite words}}.
\newblock \bibinfo{note}{CNET}.

\bibitemdeclare{article}{Muller+Schupp/95/Simulating}
\bibitem{Muller+Schupp/95/Simulating}
\bibinfo{author}{David~E. \surnamestart Muller\surnameend} \&
  \bibinfo{author}{Paul~E. \surnamestart Schupp\surnameend}
  (\bibinfo{year}{1995}): \emph{\bibinfo{title}{Simulating alternating tree
  automata by nondeterministic automata: New results and new proofs of the
  theorems of Rabin, McNaughton and Safra}}.
\newblock {\sl \bibinfo{journal}{Theoretical Computer Science}}
  \bibinfo{volume}{141}(\bibinfo{number}{1-2}), pp. \bibinfo{pages}{69--107},
  \doi{10.1016/0304-3975(94)00214-4}.

\bibitemdeclare{inproceedings}{Piterman/06/From}
\bibitem{Piterman/06/From}
\bibinfo{author}{N.~\surnamestart Piterman\surnameend} (\bibinfo{year}{2006}):
  \emph{\bibinfo{title}{From Nondeterministic Buchi and Streett Automata to
  Deterministic Parity Automata}}.
\newblock In: {\sl \bibinfo{booktitle}{21st Annual IEEE Symposium on Logic in
  Computer Science}}, pp. \bibinfo{pages}{255--264},
  \doi{10.1109/LICS.2006.28}.

\bibitemdeclare{article}{Rabin+Scott/59/Finite}
\bibitem{Rabin+Scott/59/Finite}
\bibinfo{author}{M.~O. \surnamestart Rabin\surnameend} \&
  \bibinfo{author}{D.~\surnamestart Scott\surnameend} (\bibinfo{year}{1959}):
  \emph{\bibinfo{title}{Finite automata and their decision problems}}.
\newblock {\sl \bibinfo{journal}{IBM Journal of Research and Development}}
  \bibinfo{volume}{3}, pp. \bibinfo{pages}{114--125}, \doi{10.1147/rd.32.0114}.

\bibitemdeclare{inproceedings}{Safra/88/Complexity}
\bibitem{Safra/88/Complexity}
\bibinfo{author}{S.~\surnamestart Safra\surnameend} (\bibinfo{year}{1988}):
  \emph{\bibinfo{title}{On the complexity of omega -automata}}.
\newblock In: {\sl \bibinfo{booktitle}{Proceedings of the 29th Annual Symposium
  on Foundations of Computer Science}}, \bibinfo{publisher}{IEEE Computer
  Society}, pp. \bibinfo{pages}{319--327}, \doi{10.1109/SFCS.1988.21948}.

\bibitemdeclare{inproceedings}{Safra+Vardi/89/Omega}
\bibitem{Safra+Vardi/89/Omega}
\bibinfo{author}{S.~\surnamestart Safra\surnameend} \& \bibinfo{author}{M.~Y.
  \surnamestart Vardi\surnameend} (\bibinfo{year}{1989}):
  \emph{\bibinfo{title}{On $\omega$-automata and temporal logic}}.
\newblock In: {\sl \bibinfo{booktitle}{Proceedings of the 21st annual ACM
  symposium on Theory of computing (STOC'89)}}, \bibinfo{publisher}{ACM}, pp.
  \bibinfo{pages}{127--137}, \doi{10.1145/73007.73019}.

\bibitemdeclare{inproceedings}{Safra/92/Exponential}
\bibitem{Safra/92/Exponential}
\bibinfo{author}{Shmuel \surnamestart Safra\surnameend} (\bibinfo{year}{1992}):
  \emph{\bibinfo{title}{Exponential Determinization for omega-Automata with
  Strong-Fairness Acceptance Condition (Extended Abstract)}}.
\newblock In: {\sl \bibinfo{booktitle}{Proceedings of the 24th annual ACM
  symposium on Theory of computing (STOC'92)}}, pp. \bibinfo{pages}{275--282},
  \doi{10.1145/129712.129739}.

\bibitemdeclare{inproceedings}{Schewe/09/Buchi}
\bibitem{Schewe/09/Buchi}
\bibinfo{author}{Sven \surnamestart Schewe\surnameend} (\bibinfo{year}{2009}):
  \emph{\bibinfo{title}{B{\"u}chi Complementation Made Tight}}.
\newblock In: {\sl \bibinfo{booktitle}{Proceedings of the 26th International
  Symposium on Theoretical Aspects of Computer Science (STACS 2009)}}, pp.
  \bibinfo{pages}{661--672}, \doi{10.4230/LIPIcs.STACS.2009.1854}.

\bibitemdeclare{inproceedings}{Schewe/09/Tighter}
\bibitem{Schewe/09/Tighter}
\bibinfo{author}{Sven \surnamestart Schewe\surnameend} (\bibinfo{year}{2009}):
  \emph{\bibinfo{title}{Tighter Bounds for the Determinization of B{\"u}chi
  Automata}}.
\newblock In: {\sl \bibinfo{booktitle}{Proceedings of the 12th International
  Conference on Foundations of Software Science and Computation Structures
  (FoSSaCS 2009)}}, pp. \bibinfo{pages}{167--181}.

\bibitemdeclare{incollection}{Schwoon/2002/Determinization}
\bibitem{Schwoon/2002/Determinization}
\bibinfo{author}{Stefan \surnamestart Schwoon\surnameend}
  (\bibinfo{year}{2002}): \emph{\bibinfo{title}{Determinization and
  Complementation of Streett Automata}}.
\newblock In: {\sl \bibinfo{booktitle}{Automata Logics, and Infinite Games}},
  {\sl \bibinfo{series}{Lecture Notes in Computer Science}}
  \bibinfo{volume}{2500}, \bibinfo{publisher}{Springer Berlin / Heidelberg},
  pp. \bibinfo{pages}{257--264}, \doi{10.1007/3-540-36387-4\_5}.

\bibitemdeclare{inproceedings}{Vardi/07/Buchi}
\bibitem{Vardi/07/Buchi}
\bibinfo{author}{Moshe~Y. \surnamestart Vardi\surnameend}
  (\bibinfo{year}{2007}): \emph{\bibinfo{title}{The {B\"{u}chi} complementation
  saga}}.
\newblock In: {\sl \bibinfo{booktitle}{Proceedings of the 24th International
  Symposium on Theoretical Aspects of Computer Science (STACS 2007)}},
  \bibinfo{publisher}{Springer-Verlag}, pp. \bibinfo{pages}{12--22},
  \doi{10.1007/978-3-540-70918-3\_2}.

\bibitemdeclare{inproceedings}{Yan/06/Lower}
\bibitem{Yan/06/Lower}
\bibinfo{author}{Qiqi \surnamestart Yan\surnameend} (\bibinfo{year}{2006}):
  \emph{\bibinfo{title}{Lower Bounds for Complementation of $\omega$-Automata
  Via the Full Automata Technique}}.
\newblock In: {\sl \bibinfo{booktitle}{Proceedings of the 33rd Internatilonal
  Collogquium on Automata, Languages and Programming (ICALP 2006)}}, pp.
  \bibinfo{pages}{589--600}, \doi{10.1007/11787006\_50}.

\end{thebibliography}
\end{document}